\documentclass[11pt]{article}
\usepackage{fullpage}
\usepackage{amssymb,amsmath}
\usepackage{amsmath,amsthm,amsfonts,amssymb,amscd}

\newtheorem{lemma}{Lemma}
\newtheorem{theorem}[lemma]{Theorem}
\newtheorem{corollary}[lemma]{Corollary}
\newtheorem{proposition}[lemma]{Proposition}

\newtheorem{definition}[lemma]{Definition}

\newcommand\C{{\mathbb{C}}}

\newcommand\F{{\mathbb{F}}}

\newcommand\SL{{\mathop\textup{SL}}}
\newcommand\GL{{\mathop\textup{GL}}}
\newcommand\Sp{{\mathop\textup{Sp}}}
\newcommand\softoh{{\widetilde{O}}}
\newcommand\remove[1]{{}}

\newcommand\irr{{\mathop\textup{Irr}}}

\usepackage{pgfplots}
\pgfplotsset{height=7cm,width=10cm,compat=1.9}

\title{A new algorithm for fast generalized DFTs\footnote{A version of
  this paper appeared in SODA 2018 as \cite{HU18}.}}
\author{Chloe Ching-Yun Hsu \\ Caltech
  \and Chris Umans\thanks{Supported by NSF grant CCF-1423544 and a Simons Foundation Investigator grant.}\\ Caltech}

\begin{document}
\maketitle
\setcounter{page}{0}
\thispagestyle{empty}

\date
\abstract{We give an new arithmetic algorithm to compute the
  generalized Discrete Fourier Transform (DFT) over finite groups $G$. 
The new algorithm uses $O(|G|^{\omega/2 + o(1)})$
  operations to compute the generalized DFT over finite groups of Lie type, including the linear,
  orthogonal, and symplectic families and their variants, as well as
  all finite simple groups of Lie type. Here $\omega$ is the exponent of
  matrix multiplication, so the exponent $\omega/2$ is optimal if
  $\omega = 2$. 

Previously, ``exponent one'' algorithms were known for
  supersolvable groups and the symmetric and alternating
  groups. No exponent one algorithms were known, even under the
  assumption $\omega = 2$, 
  for families of linear groups of fixed dimension, and indeed the previous best-known algorithm for
  $\SL_2(\F_q)$ had exponent $4/3$ despite being the focus of
  significant effort. We unconditionally achieve exponent at most
  $1.19$ for this group, and exponent one if $\omega = 2$.

Our algorithm also yields an improved exponent for computing the
generalized  DFT over general finite groups $G$, which beats the
longstanding previous best upper bound, for any $\omega$. In
particular, assuming $\omega = 2$, we achieve exponent $\sqrt{2}$, while the previous best was $3/2$.}

\newpage

\section{Introduction}

Let $G$ be a finite group and let $\irr(G)$ denote a complete set of
irreducible representations. Given an element of the group algebra $c
\in \C[G]$, a generalized DFT is a linear transform that takes $c$
to \[\sum_{g \in G}c_g\cdot\bigoplus_{\rho \in \irr(G)}\rho(g).\]
This is the fundamental linear operation that maps the standard basis
for the group algebra $\C[G]$ to the “Fourier basis” of irreducible
representations of group $G$. It has applications in data analysis
\cite{R97}, as a component in other algorithms (including fast
operations on polynomials and in the Cohn-Umans matrix multiplication
algorithms), and as the basis for quantum algorithms for problems
entailing a Hidden Subgroup Problem \cite{MRsurvey}. As one varies the
underlying group $G$, the generalized DFT is a rich
source of structured linear maps, which one can hope to compute in
nearly-linear time, generalizing the famous Cooley-Tukey FFT for
cyclic groups of order $2^k$.  

We typically speak of the complexity of computing this map in the (non-uniform) arithmetic
circuit model and do not concern ourselves with {\em finding} the
irreducible representations. The trivial algorithm thus requires $O(|G|^2)$
operations. The best-known algorithm that works for general finite
groups $G$ achieves $O(|G|^{1.5})$ operations\footnote{Note
  that exercise 13.16 in \cite{BCS} claims that the exponent 1.5 can be
  reduced to $1.44$ but this seems to be an error, as
  discussed in Section \ref{sec:beth-clausen}.} assuming the exponent
of matrix multiplication is two (see Section
\ref{sec:beth-clausen}). 
For a number of special cases,
``exponent 1'' algorithms are known: for abelian groups,
the symmetric and alternating groups \cite{clausenfast}, and the so-called {\em
  supersolvable} groups \cite{Baum}. A group that has resisted such exponent
$1$ algorithms despite a significant amount of work is $\SL_2(\F_q)$,
where the best known algorithm achieves $O(|G|^{4/3})$ \cite{LR92}. This group was described as a ``particularly
interesting and thorny special case'' by Maslen, Rockmore, and Wolff
in \cite{MRW}. 

In this paper we obtain exponent one for $\SL_2(\F_q)$ under the assumption
that $\omega = 2$ ($\omega$ is the exponent of matrix
multiplication). Using the current best upper bound $\omega <
2.3729$ \cite{legall}, we obtain exponent $1.19$ for $\SL_2(\F_q)$ unconditionally,
which improves the previous $4/3$ exponent. Our new algorithm is
quite general and leads to a broad array of new results:
\begin{itemize}
\item we achieve exponent $\omega/2$
for essentially all linear groups including 
the general, orthogonal, and symplectic groups,
and their special and projective versions, and for all finite groups
with a {\em split $(B,N)$-pair}; we work out the
most common cases explicitly in this paper in Section
\ref{sec:linear-groups}.
\item we achieve exponent $\omega/2$ for all finite simple groups
  (see Theorem \ref{thm:all-finite-simple}). 
\item we achieve an exponent bound for general groups $G$ which beats the
longstanding previous best upper bound, for any $\omega$ (see Theorem
\ref{thm:general-groups}). To do this we prove a
structural result about arbitrary finite groups (Theorem \ref{thm:HK-pair}) that
relies on the Classification Theorem, which may be of independent
interest.   In
particular, assuming $\omega = 2$, we achieve exponent $\sqrt{2}$, while the previous best was $3/2$. 
\end{itemize}

\paragraph{The main idea.} At its core, the seminal Beth-Clausen
fast generalized DFT is a recursive algorithm
that computes a DFT with respect to $G$ by computing several DFTs
with respect to $H$, a subgroup of $G$. Each of the $[G:H]$ many
$H$-DFTs is lifted to $G$ and then summed together. See Corollary \ref{thm:beth-clausen-coarser}. A bottleneck in
this algorithm comes from the final summation step, which in general costs
$[G:H]|G|$. Since there are groups whose largest subgroup $H$ has
index at least $|G|^{1/2}$, exponent $3/2$ is the best possible within
this approach. Improvements have generally come from using specific
knowledge of how the induced representations from $H$ up to $G$ break
up; this can sometimes be used to circumvent the bottleneck
summation. In the case of supersolvable groups and the symmetric and
alternating groups, this has yielded exponent one algorithms
\cite{Baum, clausenfast}. In the
case of solvable groups, one can obtain exponent $\omega/2$
\cite{beth, CB}.

In this paper we devise a more general way to circumvent the bottleneck summation, which depends on the structure of the group rather than knowledge
of the representation theory. Our new recursive step permits us to
decompose $G$ via {\em two} subgroups $H$ and $K$, and
recurse on $H$ and $K$. See Theorem \ref{thm:main}. One side-effect is an alternative proof of the
$\omega/2$ exponent for solvable groups that does not require
knowledge of the representation theory of the group, in Section
\ref{sec:solvable}. Our reduction bears some similarity to the double
coset algorithm of \cite{double-coset}; a key difference seems to be the use
of fast matrix multiplication at an opportune time in the procedure.

\subsection{Past and related work}

A good description of past work in this area can be found in Section 13.5
of \cite{BCS}. The first algorithm generalizing beyond the
abelian case is due to Beth in 1984 \cite{beth}; this algorithm is
described in Section \ref{sec:beth-clausen} in a form often
credited jointly to Beth and Clausen. This algorithm was the best known
for the general case of an arbitrary finite group prior to this work. Two other milestones are the
$O(|G|\log|G|)$ algorithm for supersolvable groups due to Baum
\cite{Baum}, and the $O(|G|\log^3 |G|)$ algorithm for the symmetric group
due to Clausen \cite{clausenfast}. The latter algorithm was improved
to $O(|G|\log^2 |G|)$ by Maslen \cite{Maslen}, and very recently to {\em linear} for
the special case of $S_{n-k}$-invariant functions on $S_n$ with $n >
2k$ \cite{Clausen-Huhne}. Wreath products were studied by Rockmore \cite{wreath} who obtained exponent one algorithms in certain cases. 

In the 1990s, Maslen, Rockmore and coauthors, developed the so-called ``separation
of variables'' approach, which relies on non-trivial decompositions along
chains of subgroups via {\em Bratteli diagrams} and (again) detailed
knowledge of the representation theory of the underlying groups. There
is a rather large body of literature on this approach and it has been
applied to a wide variety of group algebras and more general algebraic
objects. For a fuller description of this approach and the results
obtained, the reader is referred to the surveys \cite{MRsurvey, MRsurvey2, Rrecent}, and the most
recent paper in this line of work \cite{MRW}.

For the present paper, important results for comparison are
the previous best known results for linear groups of various
sorts. We gather them in Figure \ref{fig:previous-best}. 
Notice that for each fixed dimension $n$, these all
  represent exponent $\alpha$ algorithms for $\alpha > 1$. Our methods
  give exponent $\omega/2$ algorithms for all of these groups, which
  translates to (the optimal) exponent $1$ if $\omega = 2$. Using the current
  best upper bounds on $\omega$ our methods
  give concrete
  improvements in small dimension, in all cases; we explicitly highlight only the case
  of $\SL_2(\F_q)$ in this paper.

\begin{figure}[t]
\begin{center}
\begin{tabular}{|l|l|l|}
	\hline

	Group $G$ & Upper bound & Reference \\
	
	\hline
        $\SL_2(\F_q)$ & $\softoh(q|G|)$ & Theorem 1.1 in \cite{LR92} \\
\hline
	$\GL_n(\F_q)$ & $\softoh(q^n|G|)$ &
                                                                     Theorem
                                           4.3 in \cite{maslennew} \\
	$\mbox{P}\Sp_{2n}(\F_q)$ & $\softoh(q^{5n-3}|G|)$ & Theorem 5.14 in \cite{MRsurvey2}\\
        $O_{2n+1}(\F_q)$ & $\softoh(q^{5n-3}|G|)$ & Theorem 5.14 in
                                              \cite{MRsurvey2}\\
	$O^+_{2n}(\F_q)$, $n \geq 4$ & $\softoh(q^{5n-6}|G|)$ & Theorem 5.14
                                              in
                                              \cite{MRsurvey2} \\
	\hline
\end{tabular}
\end{center}
\caption{Previously best known running times for the generalized DFT
  over various families of linear
  groups. In this table, the $\softoh(\cdot)$ notation hides lower
  order terms and the dependence on $n$.}
\label{fig:previous-best}
\end{figure}

\subsection{Notation and preliminaries} 
Throughout this paper we will use the phrase 

\smallskip
\centerline{``$G$ has a
generalized DFT using $O(|G|^{\alpha + \epsilon})$ operations, for all
$\epsilon > 0$''} 
\smallskip

\noindent where $G$ is a finite group and $\alpha \ge 1$ is a
real number. We mean by this that
there are {\em universal} constants $c_\epsilon$ independent of the
group $G$ under consideration so that for each $\epsilon > 0$, the
operation count is at most $c_\epsilon|G|^{\alpha + \epsilon}$. Such
an algorithm will be referred to as an ``exponent $\alpha$''
algorithm. This comports with the precise definition of the exponent
of matrix multiplication, $\omega$: that there are universal
constants $b_\epsilon$ for which $n \times n$ matrix
multiplication can be performed using at most $b_\epsilon n^{\omega +
  \epsilon}$ operations, for each $\epsilon > 0$. Indeed we will often
report our algorithms' operation counts in terms of $\omega$. In such
cases matrix multiplication is always used as a black box, so,
for example, an
operation count of $O(|G|^{\omega/2})$ should be interpreted to mean: if one uses a fast matrix multiplication algorithm with exponent $\alpha$
(which may range from $2$ to $3$), then the operation count is
$O(|G|^{\alpha/2}$). In particular, in real implementations, one
might well use standard matrix multiplication and plug in $3$ for
$\omega$ in the operation count bound.  

All logarithms are base 2. We use $\irr(G)$ to denote the complete set
of irreducible representations of $G$ being used for the DFT at
hand. In the presentation to follow, we assume the underlying field is
$\C$; however our algorithms work over any field $\F_{p^k}$ whose
characteristic $p$ does not
divide the order of the group, and for which $k$ is sufficiently large
for $\F_{p^k}$ to represent a complete set of irreducibles. 

A basic fact is that $\sum_{\rho \in \irr(G)} \dim(\rho)^2 = |G|$,
which implies that for all $\rho \in \irr(G)$, we have $\dim(\rho) \le |G|^{1/2}$. An
inequality that we use repeatedly is this one:
\begin{proposition}
\label{prop:prelim-inequality}
For any real number $\alpha  > 2$, we have
\[\sum_{\rho \in \irr(G)} O(\dim(\rho)^{\alpha}) \le O(|G|^{\alpha/2}).\]
\end{proposition} 
\begin{proof}
Set $\rho_{\max}$ to be an irrep of largest dimension. We have
\[\sum_{\rho \in \irr(G)} O(\dim(\rho)^{\alpha}) \le 
O(\dim(\rho_{\max})^{\alpha - 2})\sum_{\rho \in \irr(G)} \dim(\rho)^2 =
O(\dim(\rho_{\max})^{\alpha - 2}|G|) \le  O(|G|^{\alpha/2}),\]
where the last inequality used the fact that $\dim(\rho_{\max}) \le |G|^{1/2}$.
\end{proof}
We also need Lev's Theorem: 

\begin{theorem}[\cite{lev}]
Every finite group $G$ has a proper subgroup $H$ of order at least
$|G|^{1/2}$, unless $G$ is cyclic of prime order.
\label{thm:Lev}
\end{theorem}
This is easily seen to be tight by considering the cyclic group of
order $p^2$, for $p$ prime. 

In a few key places, we utilize the Kronecker product (or tensor product) of two
matrices $A$ and $B$, and there our convention is to name the indices of $A \otimes B$ so that
\[(A\otimes B)[(i,i'), (j,j')] = A[i,j]B[i',j'].\]

\section{The single subgroup reduction}
\label{sec:beth-clausen}
In this section we describe the recursive generalized DFT attributed to Beth and Clausen (see \cite{BCS}). Given a subgroup $H$
of a finite group $G$, this reduction computes a DFT with respect to
$G$ via DFTs with respect to $H$. Our presentation makes
use of fast matrix multiplication where possible and so the running
time will be expressed in terms of $\omega$. A key definition is that
of an {\em $H$-adapted basis} for the irreps of $G$. This is a basis
in which the restriction of each irrep of $G$ to $H$ respects the direct
sum decomposition into irreps of $H$. In concrete terms, this means
that for each irrep $\rho \in \irr(G)$, while for general $g \in G$,
$\rho(g)$ is a $\dim(\rho) \times \dim(\rho)$ matrix, for $g \in H$,
$\rho(g)$ is a block-diagonal matrix with block sizes coming from the
set $\{\dim(\sigma): \sigma \in \irr(H)\}$. 

\begin{theorem}
\label{thm:beth-clausen}
Let $G$ be a finite group and let $H$ be a subgroup. Then we can compute a DFT with respect to $G$ and
an $H$-adapted basis, at
a cost of $[G:H]$ many $H$-DFTs plus \[[G:H]|G| + [G:H]^2 \sum_{\sigma \in
  \irr(H)}O(\dim(\sigma)^{\omega + \epsilon})\]
operations, for all $\epsilon > 0$. 
\end{theorem}
\begin{proof}
Let $g_1, g_2, \ldots, g_t$ be a system of distinct right coset
representatives of $H$ in $G$, so $t = [G:H]$. Let $c$ be an
element of $\C[G]$. We can write \[c =
\sum_{g \in G}c_gg = \sum_{i =1}^t \left (\sum_{h \in
    H}c^{(i)}_h h\right ) g_i\]
for some elements $c^{(i)}  = \left (\sum_{h \in H}c^{(i)}_h h \right )\in \C[H]$. By computing an $H$-DFT for
each $i$, we obtain
\[s_i = \sum_{h \in H} c^{(i)}_h \bigoplus_{\sigma \in \irr(H)}
\sigma(h).\]
Let $\overline{s_i}$ be the lift of $s_i$ in which we repeat each
$\sigma \in \irr(h)$ as many times as it occurs in the irreps of
$G$. We notice that
\[\sum_{g \in G} c_g \bigoplus_{\rho \in \irr(G)} \rho(g) =
\sum_{i=1}^{t}\overline{s_i}\cdot \left ( \bigoplus_{\rho \in \irr(G)}
  \rho(g_i) \right ).\]
Moreover, since we are using an $H$-adapted basis, each of the $t$
matrix multiplications is the product of a block-diagonal matrix
having blocks whose dimensions are those of the irreps of $H$, with
a block diagonal matrix having blocks whose dimensions are those of the
irreps of $G$. If $n_{\sigma, \rho}$ denotes the number of occurences of
$\sigma \in \irr(H)$ in $\rho \in \irr(G)$, the cost of performing this structured matrix
multiplication is at most
\begin{eqnarray*}
\sum_{\sigma \in \irr(H)} \sum_{\rho \in \irr(G)} n_{\sigma, \rho}
O(\dim(\sigma)^{\omega +\epsilon})\frac{\dim(\rho)}{\dim(\sigma)}
& = &   \sum_{\sigma \in \irr(H)} O(\dim(\sigma)^{\omega -1+\epsilon}) \sum_{\rho \in \irr(G)} n_{\sigma,  \rho}\dim(\rho) \\
& = & \sum_{\sigma \in \irr(H)} O(\dim(\sigma)^{\omega -1+\epsilon})
      \dim(\sigma) 
[G:H] \\
& = & \sum_{\sigma \in \irr(H)}O(\dim(\sigma)^{\omega+\epsilon})[G:H]
\end{eqnarray*}
where the second-to-last equality used Frobenius reciprocity: $n_{\sigma, \rho}$ also
equals the number of times $\rho$ occurs in the induction of $\sigma$
from $H$ up to $G$, and then $\sum_\rho n_{\sigma, \rho}\dim(\rho)$ is
easily seen to be the dimension of the induced representation, which
is $\dim(\sigma)[G:H]$.
We have to do $[G:H]$ many of these structured multiplications, and
then sum them up. The summing costs $[G:H] |G|$ many operations, since
the block-diagonal matrices we are summing have, in general, $|G|$ nonzeros.
\end{proof}
We note that this final sum, which costs $|G|[G:H]$ operations, cannot
be accelerated by fast matrix multiplication, and this appears to have
been overlooked in the claim in \cite{BCS} that by using fast matrix
multiplication together with Theorem \ref{thm:Lev} one can achieve an
upper bound of $O(|G|^{1.44})$ for all finite groups $G$. Indeed when $|H|
= |G|^{1/2}$, which it may be in the worst case, the $|G|[G:H]$ term by itself is at least
$|G|^{3/2}$. Our ``double subgroup reduction'' can be seen as a means
to avoid having to directly compute this bottleneck sum. 

At the expense of a slightly coarser upper bound
we can remove the requirement of an $H$-adapted basis, which will
simplify our use of Theorem \ref{thm:beth-clausen} in
recursive algorithms later. 

\begin{corollary}
\label{thm:beth-clausen-coarser}
Let $G$ be a finite group and let $H$ be a subgroup. Then we can compute a DFT with respect to $G$ at
a cost of $[G:H]$ many $H$-DFTs plus $O([G:H]^2 |H|^{\omega/2 + \epsilon})$
operations, for all $\epsilon > 0$. 
\end{corollary}
\begin{proof}
Using Proposition \ref{prop:prelim-inequality} with $\alpha = \omega
+ \epsilon$, the cost from the statement of Theorem \ref{thm:beth-clausen} can be upper bounded by
\begin{equation}
	[G:H]|G| + [G:H]^2 \sum_{\sigma \in \irr(H)}\dim(\sigma)^{\omega+\epsilon} \leq 2[G:H]^2|H|^{\omega/2+\epsilon}.
\label{eq:beth-clausen-upper-bound}
\end{equation}
Note that in Theorem \ref{thm:beth-clausen} the DFT is with respect to an $H$-adapted basis. At a cost of 
\begin{equation}
\sum_{\rho \in \irr(G)} O(\dim(\rho)^{\omega + \epsilon}) \leq O(|G|^{\omega/2 + \epsilon})
\label{eq:change-basis}
\end{equation}
operations (again using Proposition \ref{prop:prelim-inequality} with $\alpha = \omega
+ \epsilon$), we can change an arbitrary basis to an
$H$-adapted basis, to which we apply Theorem \ref{thm:beth-clausen},
and then change back to the original basis. Both expression (\ref{eq:beth-clausen-upper-bound}) and expression
(\ref{eq:change-basis}) are upper bounded by $O([G:H]^2 |H|^{\omega/2 + \epsilon})$. 
\end{proof}

The single-subgroup reduction works best when the subgroup $H$ is
large. Lev's Theorem (Theorem \ref{thm:Lev}) guarantees a subgroup of size at
least $|G|^{1/2}$. Using this, one obtains the following recursive
algorithm, whose bound, using only that $\omega \le 3$, matches
Theorem 13.48 in the presentation
in \cite{BCS}.

\begin{theorem}
For every finite group $G$, there is an exponent $1+\omega/4$
algorithm computing the DFT with respect to $G$.
\label{thm:single-subgroup-recursive}
\end{theorem}
\begin{proof}
Fix $G$. We apply Corollary \ref{thm:beth-clausen-coarser}
recursively. 

If $G$ is a $p$-group, then we apply Theorem
\ref{thm:fast-dft-supersolvable} (actually we only need to do this
when $G$ is cyclic of prime order). If $G$ is the trivial group, then
the DFT is trivial as well. Otherwise, according to Theorem
\ref{thm:Lev}, there is a subgroup $H$ of size at least $|G|^{1/2}$,
to which we apply Corollary \ref{thm:beth-clausen-coarser}. 
	
Set $\delta= \min\{\epsilon, 0.1\}$, and give names to some constants: 
\begin{itemize}
\item Let $B_\delta$ be the constant hidden in the
  $[G:H]^2\cdot O(|H|^{\omega/2+\delta})$ notation of Corollary \ref{thm:beth-clausen-coarser}. 
\item Let $B$ be the constant hidden in the $O(|G| \log |G|)$ notation
  of Theorem \ref{thm:fast-dft-supersolvable}. 
\end{itemize} 
Let $T(n)$ denote an upper bound on the operation count of this recursive
algorithm for any group $G$ of order $n$. For each fixed $\epsilon >
0$, we will prove by induction on $n$ that, for a universal constant
$C_\epsilon$,
\[T(n) \le C_\epsilon n^{1 + \frac{\omega}{4} + \epsilon}\log^2 n.\] 
This clearly holds for the base case of a $p$-group or the trivial
group, provided $C_\epsilon > B$. 

When we apply Corollary \ref{thm:beth-clausen-coarser} recursively, the cost is
at most
\[[G:H] \cdot T(|H|) + [G:H]^2 \cdot   B_\delta|H|^{\omega/2+\delta},\]
where $|H| \ge |G|^{1/2}$. If we set $\gamma$ such that $|H| = |G|^\gamma$,
and thus $1/2 \le \gamma \le 1$, and apply the induction hypothesis, we obtain
\begin{eqnarray*}
T(n) & \le &  C_\epsilon n^{1-\gamma}n^{\gamma(1+\frac{\omega}{4} + \epsilon)}\log^2(n/2)
  +B_\delta n^{2(1 - \gamma)}n^{\gamma(\omega/2 + \delta)} \\
& < & C_\epsilon n^{1+\omega/4 + \epsilon}(\log n)(\log n - 1)
  +B_\delta n^{1+\frac{\omega}{4}+\frac{\delta}{2}} 
\end{eqnarray*}
which is at most $C_\epsilon n^{1 + \frac{\omega}{4} + \epsilon}\log^2
n$ as long as $C_\epsilon \ge B_\delta$.
\end{proof}

\section{The double subgroup reduction}
This section contains our main algorithmic result. Given two subgroups
$H,K$ of a finite group $G$, we show how to compute a DFT with
respect to $G$, via DFTs with respect to $H$ and $K$. We first show how to obtain an intermediate representation in terms of
tensor products of the irreps of $H$ and the irreps of $K$: 
\begin{lemma}
\label{lem:tensor}
Let $H$ and $K$ be subgroups of $G$ and let $c$ be an element of
$\C[G]$ supported on $HK$. Fix a way of writing $g = hk$ for each $g
\in HK$ (this is unique iff $H \cap K = \{1\}$). We can compute 
\[\sum_{g = hk \in HK} c_{g} \bigoplus_{\sigma \in \irr(H), \tau \in
  \irr(K)}\sigma(h) \otimes \tau(k),\]
by performing $|H|$ many $K$-DFTs and $|K|$ many $H$-DFTs.  
\end{lemma}
\begin{proof}
We can write \[c = \sum_{g \in G}c_gg = \sum_{h \in H} h \cdot \left (\sum_{k \in
    K}c^{(h)}_k k \right )\]
for some elements $c^{(h)} = 
\left ( \sum_{k \in K} c^{(h)}_kk \right ) \in \C[K]$. We perform $|H|$ many $K$-DFTs to
compute for each $h \in H$:
\[s_h = \sum_{k \in K} c^{(h)}_k \bigoplus_{\tau \in \irr(K)} \tau(k).\]
We use the notation $s_h[\tau, u, v]$ to refer to entry $(u,v)$ of component $\tau$
in the direct sum.
Then we perform $|K|$ many $H$-DFTs to compute for each $\tau \in
\irr(K)$ and $u, v \in [\dim(\tau)]$, 
\[t_{\tau, u, v} = \sum_{h \in H} s_h[\tau, u, v] \bigoplus_{\sigma \in
  \irr(H)} \sigma(h).\]
Note that $t_{\tau, u, v}[\sigma, x, y]$ is the $((x,u), (y, v))$
entry of 
$\sum_{h, k} c^{(h)}_k \sigma(h) \otimes \tau(k)$
and note that $c^{(h)}_k = c_{hk}$, so we have computed:
\[\sum_{h,k} c_{hk} \bigoplus_{\sigma \in \irr(H), \tau \in
  \irr(K)}\sigma(h) \otimes \tau(k)\]
as promised. 
\end{proof}

The following is an important (and known) general observation (see, e.g., Lemma
4.3.1 in \cite{HJ}):  
\begin{lemma}
\label{lem:mat-mult}
If $A$ is an $n_1 \times n_2$ matrix, $B$ is an $n_2 \times n_3$ matrix, and $C$ is an
$n_3 \times n_4$ matrix, then
the product $ABC$ can be computed by multiplying $A \otimes C^T$
(which is an $n_1n_4 \times n_2n_3$ matrix) by $B$ viewed as an
$n_2n_3$-vector. 
\end{lemma}
\begin{proof}
Observe that
\[(ABC)[i_1, i_4] = \sum_{i_2, i_3} A[i_1, i_2]B[i_2, i_3]C[i_3, i_4]\]
and
\[((A \otimes C^T)\cdot B)[(i_1,i_4)] = \sum_{i_2, i_3} (A \otimes C^T)[(i_1,
i_4), (i_2, i_3)]B[(i_2, i_3)] = \sum_{i_2, i_3} A[i_1, i_2]C[i_3,
i_4]B[i_2, i_3].\]
\end{proof}
This $n_1n_4 \times n_2n_3$-matrix-vector
multiplication costs $O(n_1n_4n_2n_3)$ operations. More importantly,
we have: 
\begin{corollary}
\label{cor:repeated-mat-mult}
If $A$ and $C$ are as above, and square (so $n_1 = n_2$ and $n_3 = n_4$), and we have several $n_2 \times n_3$
matrices, $B_1, B_2, \ldots, B_\ell$, then we can compute $AB_iC$ for
all $i$ from $A \otimes C^T$, at a cost of \[O((n_2n_3)^{\omega - 1 + \epsilon} \cdot \max \{n_2n_3, \ell\}).\]
operations, for all $\epsilon > 0$. 
\end{corollary}
\begin{proof}
Set $N = n_1n_4 = n_2n_3$. If $\ell \le N$, then this can be
accomplished with a single $N \times N$ matrix multiplication, at a cost of $O(N^{\omega +
  \epsilon})$ operations, by the definition of $\omega$. If $\ell > N$, then
this can be accomplished with $\lceil \ell/N \rceil$ many $N \times N$
matrix multiplications, at a cost of $O(\ell\cdot N^{\omega - 1 +
  \epsilon})$ operations.
\end{proof}

Now we show how to lift from the intermediate representation to the
space of irreducibles of $G$. We need some notation. For $\sigma \in \irr(H), \tau \in \irr(K), \rho
\in \irr(G)$, let $n_{\sigma, \rho}$ be the number of occurences of $\sigma$
in the restriction of $\rho$ to $H$, and let $m_{\tau, \rho}$ be the
number of occurences of $\tau$ in the restriction of $\rho$ to
$K$.

\begin{lemma}
\label{lem:lift}
There is a linear map
\[\phi_{G, H, K}:\prod_{\sigma \in \irr(H), \tau \in \irr(K)} \C^{(\dim(\sigma)
  \dim(\tau))^2} \rightarrow \prod_{\rho
  \in \irr(G)} \C^{\dim(\rho)^2}\] that maps
$\bigoplus_{\sigma \in \irr(H), \tau \in \irr(K)}\sigma(h) \otimes \tau(k)$
to 
$\bigoplus_{\rho \in \irr(G)} \rho(hk)$
for all $h \in H, k \in K$. Map $\phi_{G, H, K}$ can be computed using
\begin{eqnarray*}
\sum_{\sigma \in \irr(H), \tau \in \irr(K)} O\left
  ((\dim(\sigma)\dim(\tau))^{\omega -1 + \epsilon}\cdot \max \left \{\dim(\sigma)\dim(\tau), \sum_{\rho \in \irr{G}} n_{\sigma,
    \rho}m_{\tau,\rho} \right \}\right ) \\
+ \sum_{\rho \in
\irr(G)}O(\dim(\rho)^{\omega + \epsilon})
\end{eqnarray*}
operations, for all $\epsilon > 0$. 
\end{lemma}
\begin{proof}
Let $\irr^*(H)$ be the multiset of irreducibles of $H$ in the
multiplicities that they occur in the restrictions to $H$ of $\irr(G)$, and let $\irr^*(K)$ be
the multiset of irreducibles of $K$ in the multiplicities that they
occur in the restrictions to $K$ of $\irr(G)$.
Let $S$ be the change of basis matrix taking
$\oplus_{\sigma \in \irr^*(H)} \sigma$ to $\oplus_{\rho \in
  \irr(G)} \rho$ and
let $T$ be the change of basis matrix taking $\oplus_{\tau \in
  \irr^*(K)} \tau$ to $\oplus_{\rho \in \irr(G)} \rho$.
Then for each $h \in H, k \in K$, we have 
\[S \left (\bigoplus_{\sigma \in \irr^*(H)} \sigma(h) \right )
S^{-1}T
\left (\bigoplus_{\tau \in \irr^*(K)}\tau(k)\right ) T^{-1} = \bigoplus_{\rho \in
  \irr(G)} \rho(hk).\]
Set $M = S^{-1}T$, and consider the expression
\begin{equation}
\label{eq:big-product}
\left (\bigoplus_{\sigma \in \irr^*(H)} \sigma(h) \right ) M
\left (\bigoplus_{\tau \in \irr^*(K)}\tau(k)\right ).
\end{equation}
Note that both $M$ and the above product are block-diagonal matrices with
blocks of dimension $\dim(\rho)$ as $\rho$ runs through
$\irr(G)$. Now, for each $\rho \in \irr(G)$, a given $\sigma \in
\irr(H)$ occurs $n_{\sigma, \rho}$ times and a given $\tau \in
  \irr(K)$  occurs $m_{\tau, \rho}$ times;  therefore we are
computing $\sigma(h)B_i\tau(k)$ for $p$ distinct sub-matrices $B_i$
of $M$, where $p = \sum_{\rho
\in \irr(G)}n_{\sigma, \rho}m_{\tau, \rho}$.
By Corollary \ref{cor:repeated-mat-mult}, each such batch can be computed by taking
a product of $\sigma(h)\otimes\tau(k)^T$ with a matrix whose columns
are the $B_i$ sub-matrices, viewed as vectors. This is linear in the
entries of $\sigma(h)\otimes\tau(k)$, and costs 
\[O\left ((\dim(\sigma)\dim(\tau))^{\omega - 1 +
  \epsilon}\cdot\max\left \{\dim(\sigma)\dim(\tau), \sum_{\rho \in
  \irr(G)}n_{\sigma, \rho}m_{\tau, \rho}\right\}\right )\]
operations. 
Finally, we need to multiply (\ref{eq:big-product}) by $S$ on the left
and $T^{-1}$ on the right; both maps are linear in the entries of 
$\bigoplus_{\sigma \in \irr(H), \tau \in \irr(K)}
  \sigma(h)\otimes \tau(k),$
and as block-diagonal matrix multiplications, both cost $\sum_{\rho \in \irr(G)}O(\dim(\rho)^{\omega +
  \epsilon})$ operations.  
\end{proof}

Now we use elementary facts from representation theory to bound the complexity estimate in Lemma
\ref{lem:lift} in terms of $|H|, |K|, |G|$. 
\begin{lemma} For all finite groups $G$ and subgroups $H, K$, the expression
\begin{eqnarray*}
\sum_{\sigma \in \irr(H), \tau \in \irr(K)} O\left
  ((\dim(\sigma)\dim(\tau))^{\omega -1 + \epsilon}\cdot \max \left \{\dim(\sigma)\dim(\tau), \sum_{\rho \in \irr{G}} n_{\sigma,
    \rho}m_{\tau,\rho} \right \}\right ) \\
+ \sum_{\rho \in
\irr(G)}O(\dim(\rho)^{\omega + \epsilon})
\end{eqnarray*}
is upper bounded by 
$O((|H||K|)^{\omega/2 + \epsilon/2} + |G|^{\omega/2 + \epsilon/2}).$ 
\end{lemma}
\begin{proof}
We use only the fact that for each $\rho \in \irr(G)$,
\begin{equation}
\sum_{\sigma \in \irr(H)} \dim(\sigma)n_{\sigma, \rho} =
  \dim(\rho),
\label{eq:elem1}
\end{equation}
and similarly
\begin{equation}
\sum_{\tau \in \irr(K)} \dim(\tau)m_{\tau, \rho} =
  \dim(\rho),
\label{eq:elem2}
\end{equation}
together with the fact that the sum of the squares of the dimensions
of the irreps of a group is the order of that group (which implies
that the maximum dimension is at most the square root of the order of
the group). 

We observe that by replacing the ``max'' with addition,
\begin{align*}
\sum_{\sigma \in \irr(H), \tau \in \irr(K)} &  O\left
  ((\dim(\sigma)\dim(\tau))^{\omega - 1 + \epsilon}
\cdot \max \left \{\dim(\sigma)\dim(\tau), \sum_{\rho \in \irr{G}} n_{\sigma,
    \rho}m_{\tau,\rho} \right \}\right ) \\
\le & 
 \sum_{\sigma \in \irr(H), \tau \in \irr(K)} O\left
   ((\dim(\sigma)\dim(\tau))^{\omega - 1 + \epsilon}
\cdot \left (\dim(\sigma)\dim(\tau) + \sum_{\rho \in \irr{G}} n_{\sigma,
    \rho}m_{\tau,\rho} \right )\right )
\end{align*}
We know that  
\begin{align*}
\sum_{\sigma \in \irr(H), \tau \in \irr(K)} &
(\dim(\sigma)\dim(\tau))^{\omega - 1 + \epsilon}
\cdot\dim(\sigma)\dim(\tau) \\  & =  \left ( \sum_{\sigma \in
  \irr(H)}\dim(\sigma)^{\omega + \epsilon} \right ) \cdot \left (
\sum_{\tau \in \irr(K)}\dim(\tau)^{\omega + \epsilon} \right )  \le
(|H||K|)^{\omega/2 + \epsilon/2}.
\end{align*}
where the last inequality applied Proposition
\ref{prop:prelim-inequality} twice, with $\alpha=\omega + \epsilon$.
Also, we know that 
\begin{align*}
\sum_{\sigma \in \irr(H), \tau \in \irr(K)}  &
(\dim(\sigma)\dim(\tau))^{\omega - 1 + \epsilon} 
\cdot \left (\sum_{\rho \in \irr{G}} n_{\sigma,
    \rho}m_{\tau,\rho} \right ) \\
= & 
\sum_{\rho \in \irr(G)} \left (\sum_{\sigma \in \irr(H)}
    \dim(\sigma)^{\omega - 1 + \epsilon}n_{\sigma, \rho} \right)\cdot \left (\sum_{\tau \in \irr(K)}
    \dim(\tau)^{\omega - 1 + \epsilon}m_{\tau, \rho} \right ) \\ 
\le & \sum_{\rho \in \irr(G)} \left (|H|^{(\omega - 2 + \epsilon)/2}\cdot\sum_{\sigma \in \irr(H)}
    \dim(\sigma) n_{\sigma, \rho} \right)\cdot \left (|K|^{(\omega - 2
      + \epsilon)/2}\cdot \sum_{\tau \in \irr(K)}
    \dim(\tau)m_{\tau, \rho} \right ) \\
 = & \sum_{\rho \in \irr(G)} |H|^{(\omega - 2 +
    \epsilon)/2}|K|^{(\omega - 2
      + \epsilon)/2} \dim(\rho)^2 = (|H||K|)^{(\omega - 2 +
    \epsilon)/2} |G|
\end{align*}
where the second-to-last equality used (\ref{eq:elem1}) and (\ref{eq:elem2}). If
$|H||K| \le |G|$ then this expression is at most $|G|^{\omega/2 +
  \epsilon/2}$; if $|H||K| > |G|$ then this expression is at most
$(|H||K|)^{\omega/2 + \epsilon/2}$. Finally, we have that the final
term in the main expression, $\sum_{\rho
  \in \irr(G)} O(\dim(\rho)^{\omega + \epsilon})$, is at most  $O(|G|^{\omega/2 +
  \epsilon/2})$, by Proposition \ref{prop:prelim-inequality} with
$\alpha = \omega + \epsilon$, and the lemma follows. 
\end{proof}

Our main theorems put everything together:
\begin{theorem}
\label{thm:translated-HK}
Let $G$ be a finite group and let $H, K$ be subgroups and let $x \in
G$ be any element. Fix a way of writing $g = hk$ for each $g
\in HK$ (this is unique iff $H \cap K = \{1\}$). Let $c \in \C[G]$ be supported on
$HKx$. Then we can compute 
\[\sum_{g =hkx \in HKx} c_g \cdot \bigoplus_{\rho \in \irr(G)}\rho(g)\]
at the cost of
$|H|$ many $K$-DFTS, $|K|$ many $H$-DFTs, plus $O(|G|^{\omega/2 +
  \epsilon} + (|H||K|)^{\omega/2 + \epsilon})$ operations, for all
$\epsilon > 0$. 
\end{theorem}
\begin{proof}
Set $c'_{g} = c_{gx}$ and notice that $c'$ is supported on $HK$. Apply Lemma \ref{lem:tensor} on $c'$ to compute
\[\sum_{g =hk \in HK} c'_{g} \bigoplus_{\sigma \in \irr(H), \tau \in
  \irr(K)}\sigma(h) \otimes \tau(k).\]
Next, apply the linear map $\phi_{G, H, K}$ to obtain (by
linearity) $\sum_{g =hk \in HK} c'_{g} \bigoplus_{\rho \in \irr(G)}
\rho(hk)$, and finally, multiply by $\oplus_{\rho \in \irr(G)}
\rho(x)$ on the right, at a cost of $\sum_{\rho \in
  \irr(G)}\dim(\rho)^2 \le O(|G|^{\omega/2 + \epsilon})$
operations (by Proposition \ref{prop:prelim-inequality} with $\alpha =
\omega + \epsilon$). The result is
\[\sum_{g =hk \in HK} c'_{g} \bigoplus_{\rho \in \irr(G)}
\rho(gx) = \sum_{g' \in HKx} c_{g'} \bigoplus_{\rho \in \irr(G)}\rho(g'),\]
as promised. 
\end{proof}

By translating $HK$ around, we
cover all of $G$, leading to our main theorem:
\begin{theorem}[main]
\label{thm:main}
Let $G$ be a finite group and let $H, K$ be subgroups. Then we can compute the DFT with respect to $G$ at the cost of
$|H|$ many $K$-DFTS, $|K|$ many $H$-DFTs, plus $O(|G|^{\omega/2 +
  \epsilon} + (|H||K|)^{\omega/2 + \epsilon})$ operations, all
repeated $r = O(\frac{|G|\ln(|G|)}{|HK|})$ many times, for all
$\epsilon > 0$.  If $G = HK$, then we may take $r = 1$. 
\end{theorem}
\begin{proof}
We argue that there exist $x_1, x_2, \ldots, x_r \in G$ so that $\cup_i
HKx_i = G$. Then a $G$-DFT can be computed by applying Theorem \ref{thm:translated-HK} $r$ times
with these translations. The existence of the $x_i$ is a standard
application of the probablistic method: for randomly chosen $x_i$, the
probability $\cup_i HKx_i$ fails to contain a given $g \in G$ is $(1 -
|HK|/|G|)^r$, and the $r$ specified in the theorem statement makes
this quantity strictly less than $1/|G|$, so a union bound finishes
the argument. 
\end{proof}

\section{Exponent $\omega/2$ for finite solvable groups}
\label{sec:solvable}
We show how to derive algorithms for all solvable groups via our
reduction, matching the exponent $\omega/2$ algorithm of \cite{beth, CB}. An advantage of our approach is that we don't need to rely on
knowledge of the representation theory of $G$.  

We begin with a key definition:
\begin{definition}
A finite group $G$ is {\em supersolvable} if there is a sequence of
subgroups
\[\{1\}= G_0 \lhd G_1 \lhd G_2 \lhd \cdots \lhd G_k = G\]
such that each $G_i$ is normal in $G$, and for all $i$, $G_i/G_{i-1}$ is cyclic of prime order. 
\end{definition}
A {\em solvable} finite group $G$ is one in which the
requirement that each $G_i$ is normal in $G$ (rather than just $G_{i+1}$) is removed. An early result in the area of fast generalized DFTs was Baum's
algorithm which gives a fast DFT for all {\em supersolvable} groups.
\begin{theorem}[Baum]
\label{thm:fast-dft-supersolvable}
There is an algorithm that uses $O(|G|\log |G|)$ operations to compute
the generalized DFT over $G$, if $G$ is supersolvable. 
\end{theorem}

An important class of supersolvable groups are $p$-groups. 
Together with this fact, the result of the previous section makes it quite easy to obtain an algorithm for all solvable
groups. We need the following classical result of Hall:
\begin{theorem}[Hall]
Let $G$ be a finite solvable group of order $ab$, with $(a,b)=1$. Then
there exists a subgroup $H \subseteq G$ of order $a$. 
\end{theorem}

From this we obtain:

\begin{theorem}
Let $G$ be a finite solvable group. Then a $G$-DFT can be computed in
$O(|G|^{\omega/2 + \epsilon})$ operations, for all $\epsilon > 0$.  
\end{theorem}
\begin{proof}

Take $\delta = \epsilon / 2$. Let $A_{\delta} \ge 1$ be the constant hidden
in the $O(|G|^{\omega/2+\delta}+(|H||K|)^{\omega/2+\delta})$ notation
in Theorem \ref{thm:main}. Let $B$ be the constant in the big-oh
expression in the statement of Theorem
\ref{thm:fast-dft-supersolvable}. It suffices to prove that for any
finite group $G$ with $|G|$ having $k$ distinct prime factors, a $G$-DFT can be
computed in \[(4A_{\delta})^{\log k}|G|^{\omega/2+\delta}B\log|G|\]
operations, because for sufficiently large $G$, we have 
\[(4A_{\delta})^{\log k}B\log |G| \leq (4A_{\delta})^{\log \log
  |G|}B\log |G| \leq |G|^{\delta}.\]

The proof is by induction on the number of distinct prime factors in
the order of $G$. For the base case of $k = 1$, $G$ is a $p$-group, hence supersolvable, and we apply Theorem \ref{thm:fast-dft-supersolvable}.

Now, suppose $|G| = p_1^{a_1}\dots p_k^{a_k}$, where $p_1,\dots,p_k$
are distinct primes, then $|G| = ab$, where $a$ and $b$ each has no
more than $k/2$ distinct prime factors and $(a,b) = 1$. Applying
Hall's theorem (twice) there are subgroups $H, K$ of order $a$ and $b$
respectively. Since $(a, b) =1$, we must have $H \cap K = \{1\}$, and
then $G = HK$ because $|G|=ab$.  

We can then apply Theorem \ref{thm:main},  to reduce to the case of computing
$|H|$ many $K$-DFTs and $K$ many $H$-DFTs, at a cost of 
$2A_{\delta}|G|^{\omega/2+\delta}$ operations. But $H$ and $K$ are both
solvable, and hence by the induction hypothesis, these two sets of DFTs cost at most
\begin{eqnarray*}
& & |H|\cdot (4A_{\delta})^{\log
      (k/2)}|K|^{\omega/2+\delta}B\log|K| + |K|\cdot
    (4A_{\delta})^{\log (k/2)}|H|^{\omega/2+\delta}B\log|H|\\
& & \leq \frac{2}{4A_\delta}(4A_{\delta})^{\log
    k}|G|^{\omega/2+\delta}B\log|G|
\end{eqnarray*}
operations. Together with the $2A_{\delta}|G|^{\omega/2+\delta}$
overhead, this is no more than \[(4A_{\delta})^{\log
  k}|G|^{\omega/2+\delta}B\log|G|\] operations, as required. 
\end{proof}

\section{Exponent $\omega/2$ for finite groups of Lie type}
\label{sec:linear-groups}

One of the main payoffs of Theorem \ref{thm:main} is exponent
$\omega/2$ algorithms for finite groups of Lie type. This is because
groups of Lie type have an ``$LDU$-type'' decomposition which is
well-suited to Theorem \ref{thm:main}. We describe these
decompositions and the resulting DFT algorithms in this section.  All
of our ``$LDU$-type'' decompositions of groups of Lie type into three subgroups give rise to the following DFT algorithm:
\begin{theorem}
Let $H_1, H_2, H_3$ be subgroups of group $G$, and suppose all three
are either $p$-groups or abelian. Moreover, suppose $H_1H_2$ is a
subgroup and $H_1 \cap H_2 = \{1\}$ and $H_1H_2 \cap H_3 =
\{1\}$. Then there is a generalized DFT for $G$ that uses at most
\[O\left (|G|^{\omega/2 + \epsilon}\frac{|G|\log
    |G|}{|H_1||H_2||H_3|}\right )\]
operations, for all $\epsilon > 0$.
\label{thm:LDU}
\end{theorem}
\begin{proof}
We apply Theorem \ref{thm:main} to the pair $H_1H_2$ and $H_3$ at a
cost of $O(|G|^{\omega/2 + \epsilon})$ plus $|H_1H_2|$ many $H_3$-DFTs
and $|H_3|$ many $H_1H_2$-DFTs. This is all repeated
\[r = O\left (\frac{|G|\log |G|}{|H_1||H_2||H_3|}\right )\]
many times. The $H_3$-DFTs cost $O(|H_3|\log |H_3|)$
because $H_3$ is abelian or a $p$-group (via
Theorem \ref{thm:fast-dft-supersolvable}). We apply Theorem \ref{thm:main} once
more to $H_1, H_2$, at a cost of $O(|H_1H_2|^{\omega/2 + \epsilon})$ plus $|H_1|$ many $H_2$-DFTs
and $|H_2|$ many $H_1$-DFTs. Each $H_1$-DFT costs $O(|H_1|\log |H_1|)$
because $H_1$ is abelian or a $p$-group, and the same is true for each
$H_2$-DFT. Altogether, the cost is
\begin{eqnarray*}
r\cdot \left [O(|G|^{\omega/2 + \epsilon}) \right .& + & |H_1H_2|\cdot O(|H_3| \log
|H_3|) \\
 & + & \left .|H_3|\cdot \left (O(|H_1H_2|^{\omega/2  + \epsilon}) +
  |H_1|\cdot O(|H_2|\log|H_2|) + H_2\cdot O(|H_1 \log H_1|) \right )
\right ]
\end{eqnarray*}
which is as claimed.
\end{proof}

From Carter \cite{Carter}, we have that all finite simple groups of Lie type
(except the Tits group) have a {\em split ($B$,
$N$)-pair}, which implies the following structure: 
\[G = \sqcup_{w \in W} BwU_w\]
$B$ and $N$ are subgroups, $W$ is the Weyl group (i.e. $W = B/(B \cap N)$), and $B = UT$
with $T$ a maximal torus (hence abelian) and $U, T$ are complements in
$B$. The $U_w$ are subgroups of $U$, and $U$ is a $p$-group. This
decomposition is ``with uniqueness of expression'' which implies
that $|BwU_w| = |B||U_w|$ for each $w$.

From this description we easily have the very general result:
\begin{theorem}
\label{thm:bn-pair}
Let $G$ be a finite group with a
split ($B$, $N$)-pair, with associated Weyl group
$W$. Then there is a fast DFT over $G$ that uses
$O(|G|^{\omega/2 + \epsilon}|W|)$ operations, for all $\epsilon > 0$. 
\end{theorem}
\begin{proof}
Fix the $w$ maximizing the size of the double coset $BwU_w$, and note
that $|BU_w^w| = |BwU_w| \ge |G|/|W|$. As noted this size is $|B||U_w|$, and hence
$B \cap U_w^w = \{1\}$.  Also from the description above, $B = UT$
with $U \cap T = \{1\}$; $T$ is abelian and $U, U_w^w$ are
$p$-groups. We are then in the position to apply Theorem
\ref{thm:LDU}, which yields the claimed operation count. 
\end{proof}

As one can see from Figure \ref{table:lie-type}, for families of finite simple
groups of Lie type, the Weyl group always
has order that is $|G|^{o(1)}$, so this algorithm has exponent
$\omega/2$, which is best-possible if $\omega = 2$. 
Next, we explicitly work out the more common cases of the general linear, orthogonal,
and symplectic families, and their variants. The
overhead coming from the parameter $r$ in Theorem \ref{thm:main} in each case is somewhat smaller than the
worst-case bound of $O(|W|\log |G|)$ coming from (the very general) Theorem \ref{thm:bn-pair}; instead it
approaches $O(\log |G|)$ as the underlying field size $q$ approaches infinity.

\subsection{The groups $\GL_n(\F_q)$ and $\SL_n(\F_q)$}
\label{sec:GL}
The easiest example for applying Theorem \ref{thm:LDU} is the general linear group.
\begin{theorem}
For each $n$ and prime power $q$, there is a generalized DFT for the
group $G = \GL_n(\F_q)$ that uses
$O(|G|^{\omega/2 + \epsilon})$
operations, for all $\epsilon > 0$.
\label{thm:GL}
\end{theorem}
\begin{proof}
The three subgroups $H_1, H_2, H_3$ are the set of lower-triangular matrices with ones on the
diagonal, the set of diagonal matrices, and the set of upper-triangular matrices with ones on the
diagonal, which have sizes $q^{(n^2 -
  n)/2}, (q-1)^n$, and $q^{(n^2-n)/2}$, respectively. In the notation
of Theorem \ref{thm:LDU}, we
have
\[r = O\left (\frac{|G|\log |G|}{|H_1||H_2||H_3|}\right ) \le
O\left (\frac{q}{q-1}\right )^n(n^2 \log q)\]
which can be absorbed into the $|G|^{\epsilon}$ term.
\end{proof}
For $\SL_n(\F_q)$ the only difference is that the diagonal matrices
must have determinant one, so the size of that subgroup is
$(q-1)^{n-1}$ instead of $(q-1)^n$; the group itself is also smaller
by a factor of $q-1$. We obtain in exactly the same way as for Theorem \ref{thm:GL}:
\begin{theorem}
For each $n$ and prime power $q$, there is a generalized DFT for $G = \SL_n(\F_q)$ that uses
$O(|G|^{\omega/2 + \epsilon})$
operations, for all $\epsilon > 0$.
\label{thm:SL}
\end{theorem}
Since the two dimensional case has attracted a lot of attention, we
record that result separately, for concreteness:
\begin{theorem}
\label{thm:sl2}
For each prime power $q$, there is a generalized DFT for $G = \SL_2(\F_q)$ that uses
$O(|G|^{\omega/2 + \epsilon})$
operations, for all $\epsilon > 0$.
\end{theorem}
\begin{proof}
Let $H_1$ be the set of lower triangular matrices with ones on the
diagonal, $H_2$ be the set of diagonal matrices with determinant $1$,
and $H_3$ be the set of upper triangular matrices with ones on the
diagonal. These are all subgroups, each pairwise intersection is
$\{1\}$, and we have $H_1H_2$ is a subgroup. All three subgroups are
abelian, with orders $q, q-1$, and $q$, respectively. Since $|G| = q^3
- q$ we have in this case that $|H_1H_2||H_3| = |G|$ and hence
$H_1H_2H_3 = G$. We can perform the DFT by applying Theorem
\ref{thm:main} to $H_1H_2$ and $H_3$, and then to $H_1$ and $H_2$. The
overall cost is
\begin{eqnarray*}
O(|G|^{\omega/2 + \epsilon}) &  + &  |H_1H_2|\cdot O(|H_3| \log
|H_3|)  \\ & + & |H_3|\cdot \left (O(|H_1H_2|^{\omega/2  + \epsilon}) +
  |H_1|\cdot O(|H_2|\log|H_2|) + H_2\cdot O(|H_1 \log H_1|) \right )
\end{eqnarray*}
which simplifies to the claimed operation count. 

\end{proof}

\subsection{The symplectic groups $\Sp_{2n}(\F_q)$}
\label{sec:Sp}
A symplectic group of dimension $2n$ over $\F_q$ is the subgroup of
invertible matrices that preserve a symplectic form; all symplectic
forms are equivalent under a change of basis, so concretely we may
take $\Sp_{2n}(\F_q)$ to be the set of all matrices $A \in \GL_{2n}(\F_q)$
such that \[A^TQA=Q, \mbox{ where } Q = \left(\begin{array}{cc} 0 & J \\ -J
    & 0\end{array}\right) 
\] and $J$ is the matrix with ones on the antidiagonal. 
\begin{theorem}
For each $n$ and prime power $q$, there is a generalized DFT for $G = \Sp_{2n}(\F_q)$ that uses
$O(|G|^{\omega/2 + \epsilon})$
operations, for all $\epsilon > 0$.
\end{theorem}
\begin{proof}
Let $L, U, D$ be the lower-triangular (with ones on the diagonal),
upper-triangular (with ones on the diagonal), and diagonal subgroups
of $\GL_{2n}(\F_q)$, respectively. We view  our group $G$ as a subgroup of
$\GL_{2n}(\F_q)$ as well. It is well known that the order of $G$ is 
\[q^{n^2}\prod_{i=1}^n(q^{2i}-1) \le q^{2n^2+n}.\]

Now apply Theorem \ref{thm:LDU} with $H_1 = L \cap G, H_2 = D \cap G$
and $H_3 = U \cap G$. We note that $H_1$ and $H_3$ are $p$-groups and
$H_2$ is abelian (as before). Also, $H_1H_2$ is a subgroup, and $H_1
\cap H_2 = \{1\}$ and $H_1H_2 \cap H_3 = \{1\}$. 

It remains to bound the sizes of $H_1, H_2, H_3$. In order to lower bound the size of $H_3$, consider the following
subgroups of $\GL_{2n}(\F_q)$, 
\begin{eqnarray*}
H & = &\left \{\left (\begin{array}{c|c} I_n &  M \\ \hline 0 &
      I_n \end{array} \right): M \in \F_q^{n \times n}\right \}\\
K & = & \left \{\left (\begin{array}{c|c} A &  0 \\ \hline 0 &
      B \end{array} \right): A, B \mbox{ upper
      triangular with ones on the diagonal} \right \}.
\end{eqnarray*}
One can verify that $H \cap G$ is the subgroup in which $M$ is a
persymmetric matrix (symmetric about the anti-diagonal), and thus this
subgroup has order $q^{n(n+1)/2}$. Similarly, one can verify
that $K \cap G$ is the subgroup in which $A$ is an arbitrary
upper-triangular matrix with
ones on the diagonal and $B= J(A^T)^{-1}J$. Thus this subgroup has
order $q^{n(n-1)/2}$. We have 
\[(H \cap G)(K\cap G) \subseteq H_3\]
and so $|H_3| \ge q^{n(n+1)/2 + n(n-1)/2} = q^{n^2}$. A symmetric
argument shows that $|H_1|$ has the same order. It is also easy to
verify that $|H_2| = (q-1)^n$. 
In the notation
of Theorem \ref{thm:LDU}, we
have
\[r = O\left (\frac{|G|\log |G|}{|H_1||H_2||H_3|}\right ) \le
O\left (\frac{q}{q-1}\right )^n((n^2 + n) \log q)\]
which can be absorbed into the $|G|^{\epsilon}$ term.
\end{proof}

\subsection{The orthogonal groups $O_n(\F_q)$}
\label{sec:orthog}
An orthogonal group of dimension $n$ over $\F_q$ is a subgroup of
invertible matrices that preserve a nondegenerate symmetric quadratic
form. There are several inequivalent quadratic forms and thus several
non-isomorphic orthogonal groups. For simplicity, we work out only one
case (the ``plus type'' orthogonal group of even dimension, in odd
characteristic). A similar analysis can be easily carried
out for the other non-isomorphic orthogonal groups. In our case, concretely, we may
take $O_{n}(\F_q)$ to be the set of all matrices $A \in \GL_n(\F_q)$
such that \[A^TQA=Q, \mbox{ where } Q = \left(\begin{array}{cc} 0 & J \\ J
    & 0\end{array}\right) 
\] and $J$ is the matrix with ones on the antidiagonal. 
\begin{theorem}
For each even $n$ and odd prime power $q$, there is a generalized DFT
for $G = O_n(\F_q)$ specified via the above quadratic form, that uses
$O(|G|^{\omega/2 + \epsilon})$
operations, for all $\epsilon > 0$.
\end{theorem}
\begin{proof}
Let $L, U, D$ be the lower-triangular (with ones on the diagonal),
upper-triangular (with ones on the diagonal), and diagonal subgroups
of $\GL_{n}(\F_q)$, respectively. We view  our group $G$ as a subgroup of
$\GL_{n}(\F_q)$ as well. It is well known that the order of $G$ is at most $2q^{(n^2-n)/2}$.

Now apply Theorem \ref{thm:LDU} with $H_1 = L \cap G, H_2 = D \cap G$
and $H_3 = U \cap G$. We note that $H_1$ and $H_3$ are $p$-groups and
$H_2$ is abelian (as before). Also, $H_1H_2$ is a subgroup, and $H_1
\cap H_2 = \{1\}$ and $H_1H_2 \cap H_3 = \{1\}$. 

It remains to bound the sizes of $H_1, H_2, H_3$. In order to lower bound the size of $H_3$, first consider the following
subgroups of $\GL_{n}(\F_q)$, 
\begin{eqnarray*}
H & = & \left \{\left (\begin{array}{c|c} I_{n/2} &  M \\ \hline 0 &
      I_{n/2} \end{array} \right): M \in \F_q^{n/2 \times n/2}\right \}\\
K & = & \left \{\left (\begin{array}{c|c} A &  0 \\ \hline 0 &
      B \end{array} \right): A, B \mbox{ upper
      tri. with ones on the diagonal} \right \}.
\end{eqnarray*}
One can verify that $H \cap G$ is the subgroup in which $M$ is a ``
skew-persymmetric'' matrix (skew-symmetric about the anti-diagonal),
and thus this
subgroup has order $q^{((n/2)^2 - (n/2))/2}$. Similarly, one can verify
that $K \cap G$ is the subgroup in which $A$ is an arbitrary
upper-triangular matrix with
ones on the diagonal and $B= J(A^T)^{-1}J$. Thus this subgroup has
order $q^{((n^2)^2 - (n/2))/2}$. We have 
\[(H \cap G)(K\cap G) \subseteq H_3\]
and so $|H_3| \ge q^{(n/2)^2 - (n/2)}$. A symmetric
argument shows that $|H_1|$ has the same order. It is also easy to
verify that $|H_2| = (q-1)^{n/2}$. 
In the notation
of Theorem \ref{thm:LDU}, we
have
\[r = O\left (\frac{|G|\log |G|}{|H_1||H_2||H_3|}\right ) \le
O\left (\frac{q}{q-1}\right )^{n/2}((n^2- n)\log q/2)\]
which can be absorbed into the $|G|^{\epsilon}$ term.

\end{proof}

We note that in all of the cases just considered in Sections
\ref{sec:GL}, \ref{sec:Sp}, \ref{sec:orthog}, one obtains the same
results for the special or projective (or both) variants, by following
essentially the same argument. To obtain results for the projective
cases, we observe that quotient-ing all of the groups in our
decomposition by the center can only change the operation count by a
factor of some constant multiple of the size of the center, which in
these cases is itself a constant. 

Finally, we note that Theorem \ref{thm:bn-pair} and the surrounding discussion
imply

\begin{theorem}
Let $G$ be a finite simple group. Then there is a fast DFT over $G$
that uses $O(|G|^{\omega/2 + \epsilon})$ operations, for all $\epsilon
> 0$. 
\label{thm:all-finite-simple}
\end{theorem}
\begin{proof}
As noted in the discussion before and after Theorem \ref{thm:bn-pair}, all finite
simple groups of Lie type (except the Tits group) have a split
$(B,N)$-pair, and Weyl group of order $|G|^{o(1)}$, so Theorem \ref{thm:bn-pair}
yields exponent $\omega/2$ algorithms for these families. By the
Classification Theorem, the only
other infinite families of finite simple groups are the alternating
group and the abelian groups, both of which have exponent $1$
algorithms. The sporadic groups and the Tits group are a finite set of
exceptions that can be handled by choosing the constant in the big-oh
notation sufficiently large. 
\end{proof}

\section{A new exponent upper bound for all finite groups}

In this section we prove a structural result for all finite groups that allows us to make
use of the reduction in Theorem \ref{thm:main}. Just as Lev's theorem
regarding a large single subgroup allows one to use the single subgroup
reduction of Section \ref{sec:beth-clausen}
to obtain a non-trival upper bound for all finite groups,
the following theorem gives a {\em pair} of subgroups for use
in the reduction of Theorem
\ref{thm:main}. 

\begin{theorem}
There exists a monotone increasing function $f(x)\leq 2^{c\sqrt{\log
    x}\log\log x}$ for
a universal constant $c \ge 1$, for which the following holds: every finite group
$G$ that is not a $p$-group has proper subgroups $H, K$ satisfying
$|HK| \ge |G|/f(|G|)$.
\label{thm:HK-pair}
\end{theorem}

\begin{proof}
If G is simple then by the Classification Theorem, we have several cases:
\begin{itemize}
\item $G$ is cyclic of prime order. This case cannot arise since 
  $G$ is not a $p$-group.  
\item $G$ is the alternating group $A_n$. Then we choose $H = A_{n-1}$
  and $K = \{1\}$ and we have $|HK| \ge |G|/n$, so as long as $f(x) >
\log x$, the theorem holds.
\item $G$ is a finite group of Lie Type. Then $G$ has a $(B,N)$ pair
  (the Tits Group is an exception; it does not have a $(B,N)$
  pair, but it is a single finite group so it can be treated along
  with the sporadic groups in the next case). Let $W = N/(B \cap N)$
  be the Weyl group, and from the axioms of a $(B, N)$ pair, we have
  that the double cosets $B\overline{w}B$ with
  $w \in W$ cover $G$ (the $\overline{w}$ denotes a lift to $N
  \subseteq G$). Thus there is some double coset $B\overline{w}B$ of size at
  least $|G|/|W|$. Taking $H = B^{\overline{w}}$ and $K = B$, we see
  that $|HK| = |B\overline{w}B| \ge |G|/|W|$. Now we verify that we
  can choose $f$ so that for each of the families in Figure
  \ref{table:lie-type}, $f(|G|) > |W|$.
\item $G$ is one of the sporadic groups. Let $C$ be the largest order
  of a sporadic group. Then by choosing $f(x)>C$, the theorem holds
  for $H = K = \{1\}$ in this case.
\end{itemize}
If $G$ is not simple, then let $N$ be a maximal normal subgroup of
$G$, so that $G/N$ is simple. We have two cases:
\begin{itemize}
\item $G/N$ is a $p$-group. Since $G$ is not a $p$-group, we have that
  $|G| = mp^k$ for $m > 1$ and $(m,p) = 1$. Let $P$ be a $p$-Sylow
  subgroup of $G$. Then $|P| = p^k$, and $|N| = mp^{k'}$ for some $k'
  < k$. Then $NP = G$ and both $N$ and $P$ are proper subgroups.
\item $G/N$ is a simple group that is not a $p$-group. Then apply the
  previous case analysis for simple groups to obtain $H/N, K/N$,
  proper subgroups of $G/N$ for which $|(H/N)(K/N)| \ge
  |G/N|/f(|G/N|)$. But then $H,K$ are proper subgroups of $G$ and
  \[|HK| = |(H/N)(K/N)||N| \ge
  |G/N||N|/f(|G/N|) = |G|/f(|G/N|) \ge |G|/f(|G|),\]
where the last inequality used the monotonicity of $f$.
\end{itemize}
\end{proof}

\begin{figure}\begin{center}
\begin{tabular}{|l|l|l|l|}
\hline
Name & Family & $|W|$ & $|G|$ \\
\hline
Chevalley & $A_\ell(q)$ & $(\ell + 1)!$ & $q^{\Theta(\ell^2)}$ \\
&$B_\ell(q)$ & $2^\ell\ell!$ & $q^{\Theta(\ell^2)}$ \\
&$C_\ell(q)$ & $2^\ell\ell!$ & $q^{\Theta(\ell^2)}$ \\
&$D_\ell(q)$ & $2^{\ell-1}\ell!$ & $q^{\Theta(\ell^2)}$ \\
\hline
Exceptional & $E_6(q)$ & $O(1)$ & $q^{\Theta(1)}$ \\
Chevalley & $E_7(q)$ & $O(1)$ & $q^{\Theta(1)}$ \\
& $E_8(q)$ & $O(1)$ & $q^{\Theta(1)}$ \\
& $F_4(q)$ & $O(1)$ & $q^{\Theta(1)}$ \\
& $G_2(q)$ & $O(1)$ & $q^{\Theta(1)}$ \\
\hline
Steinberg & ${}^2A_\ell(q^2)$ & $2^{\lceil \ell/2 \rceil}\lceil \ell/2\rceil!$ & $q^{\Theta(\ell^2)}$
\\
& ${}^2D_\ell(q^2)$ & $2^{\ell-1}(\ell-1)!$ & $q^{\Theta(\ell^2)}$\\
& ${}^2E_6(q^2)$ &  $O(1)$ & $q^{\Theta(1)}$ \\
&${}^3D_4(q^3)$ &  $O(1)$ & $q^{\Theta(1)}$\\
\hline
Suzuki & ${}^2B_2(q)$, $q = 2^{2n+1}$  &  $O(1)$ & $q^{\Theta(1)}$ \\
\hline
Ree & ${}^2F_4(q)$, $q = 3^{2n+1}$  &  $O(1)$ & $q^{\Theta(1)}$ \\
& ${}^2G_2(q)$, $q = 3^{2n+1}$  &  $O(1)$ & $q^{\Theta(1)}$ \\
\hline
\end{tabular}
\end{center}
\caption{Families of finite groups $G$ of Lie type, together with
  the size of their associated Weyl group $W$. These include all
  simple finite groups other than cyclic groups, the alternating groups, the 26
  sporadic groups, and the Tits group. See \cite{lev,wiki} for sources.}
\label{table:lie-type}
\end{figure}

\begin{figure}
\begin{center}
\begin{tikzpicture}[
declare function={
    func(\x)= (\x > 2.562) * (4 * \x / (4 + \x))   +
              (\x <= 2.562) * ((\x - 2 + sqrt(\x^2-4*\x+36)) / 4)
   ;
  }
]
\begin{axis}[
    axis lines = left,
    xlabel = $\omega$,
    ylabel = {DFT Exponent},
    ymin = 0.9,
    legend style={at={(1.5,0.5)},anchor=north},
]
\addplot [
    domain=2:3, 
    samples=100, 
    color=blue,
    ]
    {1 + x / 4};
\addlegendentry{Previous Best}

\addplot [
    domain=2:3, 
    samples=100, 
    color=red,
]
{func(x)};
\addlegendentry{Our Result}

\addplot [
    domain=2:3, 
    samples=100, 
    color=black,
    ]
    {x / 2};
\addlegendentry{Conjectured Optimal}

\addplot [
    domain=2:3, 
    samples=100, 
    color=green,
    ]
    {1};
\addlegendentry{Lower Bound}
 
\draw [stealth-,thick] (axis cs:2.562,1.55) -- (axis cs:2.562,1.45) node[below]{$\omega=2.562$} ; 
 
\end{axis}
\end{tikzpicture}
\end{center}
\caption{Upper bound in Theorem \ref{thm:general-groups} as a function of
  $\omega$. The previous best bound is from Theorem
  \ref{thm:single-subgroup-recursive}. Assuming that some dependence
  on fast matrix multiplication is necessary, $w/2$ is a reasonable
  conjecture for the optimal dependence. Exponent one is of course a
  trivial lower bound.}
\label{fig:bounds}
\end{figure}
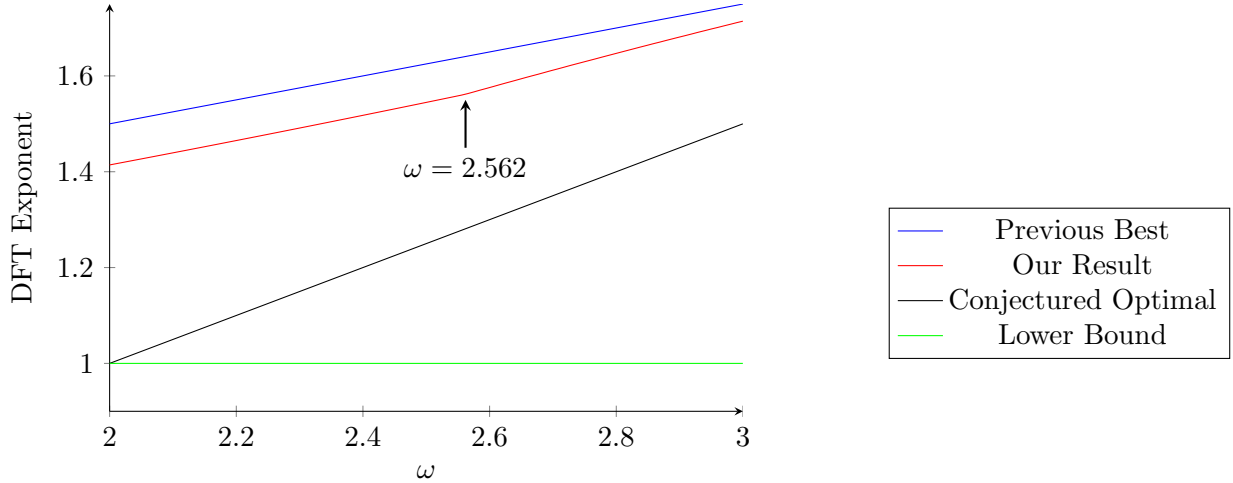

Now we can use this theorem in a recursive algorithm that switches
between the single subgroup reduction and the double subgroup
reduction, as follows:

\begin{theorem}
For every finite group $G$, there is an exponent
$\frac{\omega-2+\sqrt{\omega^2-4\omega+36}}{4}$ algorithm computing
the DFT with respect to $G$ when $\omega \leq \frac{1+\sqrt{17}}{2}$,
or exponent $\frac{4\omega}{4+\omega}$ when $\omega \geq
\frac{1+\sqrt{17}}{2}$. In particular, when $\omega = 2$, the exponent
is $\sqrt{2}$.
\label{thm:general-groups}
\end{theorem}
To visualize these bounds, refer to Figure \ref{fig:bounds}. 

\begin{proof}
We describe our general strategy before formally analyzing the
complexity. For each possible value of $\omega$, we pick a threshold
$\beta$ as a function of $\omega$. This threshold will be used to
switch between the single subgroup and the double subgroup reductions.
	
Fix $G$. Consider the following recursive algorithm. If $G$ is a
$p$-group, then we apply Theorem \ref{thm:fast-dft-supersolvable}. If
$G$ is the trivial group, then the DFT is trivial as well. Otherwise,
let $H, K$ be the subgroups guaranteed by Theorem
\ref{thm:HK-pair}. If $|H|, |K|$ are both at most $|G|^{\beta}$, then
we apply Theorem \ref{thm:main} (the double subgroup
reduction). Otherwise one of $H, K$ has size at least $|G|^{\beta}$
(without loss of generality, assume it's $H$) and we apply Corollary \ref{thm:beth-clausen-coarser}
(the single subgroup reduction).
	
Let us now analyze the operation count in terms of $\beta$. After this analysis, we'll pick the optimal $\beta$ for each $\omega$ to minimize the operation count.
	
For this purpose, set $\delta= \min\{\epsilon, 0.1, \frac{0.1\epsilon}{\beta}\}$, and give names to some constants: 
\begin{itemize}
\item Let $A_\delta$ be the constant hidden in the $O(|G|^{\omega/2 +
    \delta} + (|H||K|)^{\omega/2 + \delta})$ notation of Theorem
  \ref{thm:main}. 
\item Let $B_\delta$ be the constant hidden in the $[G:H]^2\cdot
  O(|H|^{\omega/2+\delta})$ notation of Corollary \ref{thm:beth-clausen-coarser}. 
\item Let $B$ be the constant hidden in the $O(|G| \log |G|)$ notation
  of Theorem \ref{thm:fast-dft-supersolvable}.
\end{itemize} 

Let $T(n)$ denote an upper bound on the running time of this recursive
algorithm for any group $G$ of order $n$. For each fixed $\epsilon >
0$, we will prove by induction on $n$ that, for a universal constant
$C_\epsilon$,
\[T(n) \le C_\epsilon n^{\alpha + \epsilon}\log^2 n,\] where $\alpha$ is determined by $\beta$ and $\omega$.
This clearly holds for the base case of a $p$-group or the trivial
group, provided $C_\epsilon > B$ and $\alpha \geq 1$. 

By choosing $C_\epsilon$
sufficiently large, we may assume that $|G|$ is at least some fixed
constant size, and hence we may assume that $2^{c\sqrt{\log |G|}\log\log
   | G|}\cdot O(\log |G|)$ term in the notation of Theorem
  \ref{thm:main} is bounded above by $|G|^{\epsilon/10}$.

In the case that we apply Theorem \ref{thm:main}, the cost is at most
\[\left (|H| \cdot T(|K|) + |K| \cdot T(|H|) +
  A_\delta(|H||K|)^{\omega/2+\delta}\right) \cdot |G|^{\epsilon/10},\]
where $|H|, |K| \le |G|^{\beta}$. Applying the induction
hypothesis, we obtain:
\begin{eqnarray*}
T(n) & \le & 2C_\epsilon \Big(n^{\beta}n^{\beta(\alpha +
    \epsilon)}\log^2(n^{\beta}) + A_\delta n^{2\beta(\omega/2
             + \delta)} \Big) \cdot n^{\epsilon/10} \\
& \le & (2C_\epsilon\beta^2 + A_\delta)\cdot n^{\max(\beta+\beta\alpha+\beta\epsilon,\omega\beta+2\beta\delta) + \frac{\epsilon}{10}}\log^2 n
\end{eqnarray*}
which can be bounded above by $C_\epsilon n^{\alpha + \epsilon}\log^2 n$ as long as the following constraints are satisfied:
\begin{itemize}
	\item $\beta < \frac{\sqrt{2}}{2} \approx 0.707$;
	\item $\alpha \geq \max(\frac{\beta}{1-\beta}, \omega\beta)$;
	\item $C_\epsilon > \frac{A_\delta}{1-2\beta^2}$.
\end{itemize} 

In the case that we apply Corollary \ref{thm:beth-clausen-coarser}, the cost is
at most 
\[[G:H] \cdot T(|H|) + [G:H]^2 \cdot   B_\delta|H|^{\omega/2+\delta},\]
where $|H| \ge |G|^{\beta}$ and hence $[G:H] \le
|G|^{1-\beta}$. If we set $\gamma$ such that $|H| = |G|^\gamma$,
and thus $\beta \le \gamma \le 1$, and apply the induction hypothesis, we obtain,
\begin{eqnarray*}
T(n) & \le &  C_\epsilon n^{1-\gamma}n^{\gamma(\alpha + \epsilon)}\log^2(n/2)
  +B_\delta n^{2(1 - \gamma)}n^{\gamma(\omega/2 + \delta)} \\
& < & C_\epsilon n^{\alpha + \epsilon}(\log n)(\log n - 1)
  +B_\delta n^{2-(2-\omega/2)\beta+\delta} 
\end{eqnarray*}
which is at most $C_\epsilon n^{\alpha + \epsilon}\log^2 n$ as long as the following constraints are satisfied:
\begin{itemize}
	\item $\alpha \geq 2-(2-\omega/2)\beta$;
	\item $C_\epsilon \ge B_\delta$.
\end{itemize}

To recap, the above induction proof holds when 
\[\alpha = \max(\frac{\beta}{1-\beta}, \omega\beta, 2-(2-\frac{\omega}{2})\beta), \text{ and } \beta < \frac{\sqrt{2}}{2} \approx 0.707.\]

Now we solve for the optimal $\beta$ for each fixed $\omega$.

When $\omega \geq \frac{1+\sqrt{17}}{2} \approx 2.562$, the optimal is \[\beta^* = \frac{4}{4+\omega}, \alpha^* = \frac{4\omega}{4+\omega}.\]
When $\omega \leq \frac{1+\sqrt{17}}{2} \approx 2.562$, the optimal is \[\beta^* = \frac{10-\omega-\sqrt{\omega^2-4\omega+36}}{2(4-\omega)}, \alpha^* = \frac{\omega-2+\sqrt{\omega^2-4\omega+36}}{4}.\]
\end{proof}

\section{Conclusions}

There are two significant open problems that naturally follow from the
results in this paper. First, can one obtain exponent $\omega/2$
algorithms for all finite groups? This might be possible by proving a
more sophisticated version of Theorem \ref{thm:HK-pair}, which, for
example, manages to upper bound $|H \cap K|$. Also of interest would
be a proof of Theorem \ref{thm:HK-pair} that does not need the
Classification Theorem. 

A second question is whether the dependence on $\omega$ can be removed. Alternatively, can
one show that a running time that depends on $\omega$ is necessary by
showing that an exponent one DFT for a certain family of groups would
imply $\omega = 2$?

\paragraph{Acknowledgements.} We thank the SODA 2018 referees for
their careful reading of this paper and many useful comments.

\end{document}